%% file: main.tex
\documentclass[conference,letterpaper]{IEEEtran}

\usepackage[utf8]{inputenc} 
\usepackage[T1]{fontenc}
\usepackage{url}
\usepackage{ifthen}
\usepackage{cite}
\usepackage[cmex10]{amsmath} 

\input{preamble.tex}
\begin{document}
\title{Measuring the Redundancy of Information from a Source Failure Perspective} 


\author{%
  \IEEEauthorblockN{Jesse Milzman}
  \IEEEauthorblockA{DEVCOM Army Research Laboratory\\
                    Adelphi, MD, USA\\
                    Email: jesse.m.milzman.civ@army.mil}
}


\maketitle

\begin{abstract}
In this paper, we define a new measure of the redundancy of information from a fault tolerance perspective.
The partial information decomposition (PID) emerged last decade as a framework for decomposing the multi-source mutual information $I(T;X_1, ..., X_n)$ into atoms of redundant, synergistic, and unique information.
It built upon the notion of redundancy/synergy from McGill's interaction information 
\cite{mcgill1954}.
Separately, the redundancy of system components has served as a principle of fault tolerant engineering, for sensing, routing, and control applications.
Here, redundancy is understood as the level of duplication necessary for the fault tolerant performance of a system.
With these two perspectives in mind, we propose a new PID-based measure of redundancy $\Ift$, based upon the presupposition that redundant information is robust to individual source failures.
We demonstrate that this new measure satisfies the common PID axioms from \cite{williams2010nonnegative}.
In order to do so, we establish an order-reversing correspondence between collections of source-fallible instantiations of a system, on the one hand, and the PID lattice from \cite{williams2010nonnegative}, on the other.
\end{abstract}

\begin{IEEEkeywords}
redundancy, partial information decomposition,  fault tolerance, interaction information
\end{IEEEkeywords}

\input{Sections/introduction}

\input{Sections/example}

\input{Sections/pid}

\input{Sections/ift}





\bibliographystyle{ieeetr}
\bibliography{main}

\end{document}

%% file: preamble.tex
\usepackage{amsmath,amsfonts,amssymb,amsthm}
\usepackage[utf8]{inputenc}
\usepackage{verbatim}
\usepackage{graphicx}
\usepackage{amsfonts}
\usepackage{bbm}
\usepackage{caption}

\usepackage{amsmath,amsfonts,amssymb}
\usepackage{graphicx}
\usepackage[colorlinks=true, allcolors=blue]{hyperref}

\usepackage{bm}

\newtheorem{thm}{Theorem}
\newtheorem{corollary}{Corollary}

\newtheorem{proposition}{Proposition}
\newtheorem{definition}{Definition}

\newcommand{\baseSystem}{\mathfrak{s}}
\newcommand{\sfs}{\mathfrak{f}}
\newcommand{\SFS}{\mathfrak{F}}
\newcommand{\SFSRS}{\mathfrak{R}}

\newcommand{\Ift}{I_{\texttt{ft}}}
\newcommand{\Imin}{I_{\texttt{min}}}
\newcommand{\Ibroja}{I_{\texttt{BROJA}}}

\DeclareMathOperator{\ind}{ind}

\newcommand{\Ired}{I_{\cap}}

\usepackage{tikz}
\usetikzlibrary{shapes.geometric, arrows,decorations.pathreplacing}
\tikzstyle{pose} = [circle, text centered, draw=black, fill=orange!30, minimum size = 1.25cm]
\tikzstyle{spur} = [circle, text centered, draw=black, fill=violet!30, minimum size = 1.25cm]
\tikzstyle{factor} = [rectangle, text centered, draw=black, fill=blue!30]
\tikzstyle{rendezFactor} = [rectangle, text centered, draw=black, fill=violet!30]
\tikzstyle{rendezFactorEliminated} = [rectangle, text centered, draw=black!40, text = black!40, fill=violet!10]
\tikzstyle{factorNew} = [rectangle, text centered, draw=black, fill=green!30]
\tikzstyle{factorNewEliminated} = [rectangle, text centered, draw=black!40, text=black!40, fill=green!10]

\tikzstyle{poseEliminated} = [circle, text centered, draw=black!40, text=black!40, fill=orange!10]
\tikzstyle{spurEliminated} = [circle, text centered, draw=black!40, text=black!40, fill=violet!10, minimum size = 1.25cm]
\tikzstyle{factorEliminated} = [rectangle, text centered, draw=black!40, text=black!40, fill=blue!10]

\tikzstyle{intra} = = [thick,->,>=stealth]
\tikzstyle{inter} = = [dashed,-, very thick,draw=blue]
\tikzstyle{factorLine} = = [thick,-]
\tikzstyle{factorLineEliminated} = = [thick,-,draw=black!40]
\tikzstyle{factorLineNew} == [thick, green, -]
\tikzstyle{BNarrow} = = [dashed, blue, ->]
\tikzstyle{interrobotComms} = [dashed,->,draw=violet,thick]

\usepackage{todonotes}
\usepackage[most]{tcolorbox}

\usepackage{dirtytalk}

\usepackage[framemethod=TikZ]{mdframed}
\newcounter{assumption}[section]
\renewcommand{\theassumption}{\arabic{section}.\arabic{assumption}}

\newcounter{metric}[section]
\renewcommand{\themetric}{\arabic{section}.\arabic{metric}}


%% file: Sections/introduction.tex
\section{Introduction}
There are two notions of redundancy explored in this paper.
Within the fault tolerance literature, redundancy describes the desirable property of a system to be able to replicate its component services in the event of a fault \cite{rullo2019redundancy}.
This redundancy is typically implemented through physical or logical copies of the system's fallible components.
Fault tolerance describes the system's ability to continue to operate despite faults.
Classic results from the field prescribe the minimal level of redundancy needed to withstand a given number of faults, for a given task and fault type.
For instance, one classic result considers $n$ independent sensors measuring the same continuous value within a specified interval \cite{marzullo1990tolerating}.
Depending on the nature of the faults, it was determined that the centralized algorithm provided could tolerate up to $f={(n-1)/k}$ faults, where $k=1,2,3$ for different fault types.
Besides sensing, similar results exist for routing and control \cite{rullo2019redundancy}.
Redundancy, in the form of duplicated information assets, provides gaurantees on system performance up to a certain level of failure.

Separately, there is a recurrent interest in information theory to measure the redundant or common information content among multiple variables.
Both total correlation \cite{Watanabe1960} and interaction information \cite{mcgill1954}, for instance, are different extensions to mutual information, and can both be considered as measures of redundancy or common randomness among many variables.
When one of the variables is designated as a target, this question of redundancy can be interpreted as asking for the interactions among multiple sources regarding the target (see McGill's original \cite{mcgill1954}).
For three variables $X,Y,Z$, interaction information takes the form:
\begin{equation}
    \label{eq:intxn_info}
    I(X ; Y ; Z) = I(X : Y | Z) - I(X; Y)
\end{equation}
The signed nature of the interaction suggests a rich structure: a positive interaction indicates synergy, while a negative one indicates redundancy among the variables.
Interaction information found favor in the theoretical neuroscience community in the 90s/00s, both to measure synergy \cite{Schneidman11539,timme2014synergy} and redundancy\cite{chechik2001group} (see also \cite{timme2014synergy}).

The partial information decomposition (PID) was introduced in the last decade to separate redundancy and synergy as distinct informational quantities \cite{williams2010nonnegative}.
Dispensing with the target-source symmetry of interaction information, PID seeks to decompose the information between a collection of source variables $X_1, X_2, ..., X_n$ and a target $T$ into those components that can be uniquely, redundantly, and/or synergistically attributed to each source.
For instance, when $n=2$, the so-called bivariate PID decomposes $I(T;X_1,X_2)$ into four information atoms $R+S+U_1+U_2$, satisfying:
\begin{subequations}
\label{eqs:bivariate_PID}
\begin{gather}
    I(T;X_1) = R + U_1 \\
    I(T;X_2) = R + U_2 \\
    I(T;X_1,X_2) = R +  U_1 + U_2 + S
\end{gather}    
\end{subequations}
These equations are known as the PID
equations, and these atoms are referred to as redundant ($R$), unique ($U_1$, $U_2$), and synergistic ($S$) information.
By applying the chain rule for mutual information, from (\ref{eq:intxn_info}) and (\ref{eqs:bivariate_PID}a-c) we have that
\begin{equation}
    \label{eq:intxn_info_PID}
    I(T; X_1 ; X_2) = S - R
\end{equation}
and thus recover the traditional interpretation of interaction information as measuring a synergy/redundancy trade-off.
As (\ref{eq:intxn_info_PID}) is an underdetermined system, one of the atoms needs to be specified in order to provide a proper PID definition.
Many different PIDs have been proposed, satisfying different, often mutually exclusive axioms \cite{williams2010nonnegative, broja2014,finn2018, lizier2018information}.

PID has drawn great interest from researchers at the intersection of neuroscience and complex systems \cite{lizier2018information}.
However, its application outside of biological and social science is less common, largely limited to theoretical machine learning \cite{ehrlich2023measure,liang2023quantifying}.
To our knowledge, there is no technical work examining what PID might be able to suggest about an information system's robustness from a fault tolerance perspective.
This work addresses itself to that question.

In this paper, we will develop a novel PID redundancy function $\Ift$ that quantifies the minimal average information in the presence of crash faults --- i.e. faults in which a source unambiguously fails and provides a null signal.
We begin by presenting a simple distribution in Section~\ref{sec:example} that demonstrates the type of redundancy we would like to capture.
In Sec.~\ref{sec:PID}, we review the fundamentals of the partial information decomposition, and introduce relevant notation and terminology.
In Sec.~\ref{sec:Ift}, we introduce our fault tolerance-inspired PID redudndancy function $\Ift$, and examine its mathematical properties.
In the course of demonstrating that $\Ift$ satisfies the common PID axioms from \cite{williams2010nonnegative,broja2014},
we show an insightful result: there is a order-reversing correspondence between the collections of sources on the PID lattice, over which PID redundancy is measured, and collections of source-fallible realizations of the system --- i.e. garbled, fault-prone copies of $X_1, ..., X_n$.

%% file: Sections/example.tex
\section{Motivating Example}
\label{sec:example}

Consider a scenario in which we need to know a one-bit variable $T$ of great importance --- say, who is entering a building tonight, an invited or an uninvited guest.
The building has two entrances they may enter from: a front and a back.
In ideal circumstances, both are monitored by a security camera, and so we have source variables $X_1$ and $X_2$ for the front and back cameras, respectively.
We code the meanings derived from each video feed using the alphabet $\mathcal{A}_{X_i} = \{ 1, 2, 3\}$, where $1$ means that we observe no activity at that entrance, while $2$ and $3$ signify the invited and uninvited guest entering, respectively.
Our target variable only takes values $2$ and $3$, since we have prior information that some visitor will be entering the building.
Assuming our observation of the entrances is perfect and the visitor cannot use both entrances, the distribution for the system $(X_1,X_2,T)$ is given by Table~(\ref{table:leading_example}A).

It is clear that, in this ideal scenario, $I(T;X_1,X_2) = H(X_1) = 1$, i.e. $(X_1,X_2)$ contain all the information about $T$.
Moreover, each source $X_i$ gives this full bit of information half of the time, i.e. $I(T; X_i | X_i \neq 1) = 1$.
However, when $X_i=1$, we have gain no information from it, as $H(T | X_i=1) = H(T) = 1$.
Thus each pairwise MI is given as $I(T;X_i) = 1/2$.

\begin{table}[h]
\centering
{\Large \textbf{A}}
\vspace{6pt}

\begin{tabular}{ |p{1cm}|p{1cm}|p{1cm}||p{1cm}|  }
 \hline
    $X_1$ & $X_2$ & $T$ & $p$ \\
 \hline
 1 & 2 & 2 & $1/4$ \\
 1 & 3 & 3 & $1/4$ \\
 2 & 1 & 2 & $1/4$ \\
 3 & 1 & 3 & $1/4$ \\
 \hline
\end{tabular}

\vspace{12pt}

{\Large \textbf{B}}
\vspace{6pt}

\begin{tabular}{ |p{1cm}|p{1cm}|p{1cm}||p{1cm}|  }
 \hline
 & & &
 \\[-6pt]
    $\tilde{X}_1$ &  $\tilde{X}_2$ & $T$ & $p$ \\
 \hline
 1 & 0 & 2 & $1/4$ \\
 1 & 0 & 3 & $1/4$ \\
 0 & 1 & 2 & $1/4$ \\
 0 & 1 & 3 & $1/4$ \\
 \hline
\end{tabular}

\vspace{12pt}

\caption{
(\textbf{A}) Our leading example, a trivariate, $n=2$ system with a 1-bit target. As can be seen, each predictor provides 1 bit of information half the time.
(\textbf{B}) A garbled, source-fallible copy of the same system, where the fallible sources provide no information about $T$.
This demonstrates the kind of circumstance that fault tolerant engineering aims to control for. 
}
\label{table:leading_example}
\end{table}

How, then, should we decompose this bit of mutual information into $R+U_1+U_2+S$, the bivariate PID of this system (\ref{eqs:bivariate_PID}a-c)?
Let's assume that we want each atom to be non-negative, which would exclude several recent PIDs (e.g. \cite{finn2018,ince2017measuring,makkeh2021introducing}).
Since the bivariate PID has one degree of freedom, we may consider the edge cases, corresponding to $R=1/2$ and $R=0$, respectively.

First, we might have $R=S=1/2$, which would in turn imply $U_1=U_2=0$.
This decomposition would indicate that, on average, half a bit of information is redundant between the variables, and half a bit is synergistic.
This is the PID that would be given by both the original $\Imin$ redundancy-based PID from \cite{williams2010nonnegative} and the popular $\Ibroja$ PID from \cite{broja2014}.
However, does this decomposition make sense?
In every realization $(x_1,y_1,t)$, only one of our two predictors is giving any pointwise information regarding the target, while the other is giving none.
For instance, for the first two rows of Table~\ref{table:leading_example}A,
\begin{align}
    \label{eq:example.pwMI.1}
    \log \frac{p(x_1,t)}{p(x_1)p(t)} &= 0\\
    \label{eq:example.pwMI.2}
    \log \frac{p(x_2,t)}{p(x_2)p(t)} &= 1
\end{align}
Thus, it is unclear in what sense there is redundancy (or synergy) between the information provided pointwise by these outcomes --- at least if we take redundancy to mean information that can be provided reliably by either variable in the event that the other source fails.

Let us instead consider the other option, setting $R=S=0$, which would in turn give us $U_1=U_2=1/2$.
This tells us that the information provided by $(X_1,X_2)$ about $T$ can be divided into unique contributions from each, and that none of the information can be considered redundantly provided by both, or synergistically provided by their combination.
One possible intuitive justification for this allocation could be that the supports of the pointwise informations given by the left-hand sides of Eqs~(\ref{eq:example.pwMI.1})-(\ref{eq:example.pwMI.2}) are mutually exclusive: our sample space is partitionable between the events ``only $X_1$ is informative'' and ``only $X_2$ is informative.''

We favor this latter assignment.
Taking our inspiration from the notion of redundancy as used in the fault tolerance literature, we conceptualize `redundant information' as the expected bits that one can guarantee if at least one of the sources is available in every realization --- or, equivalently, if we allow that all but one source may fail arbitrarily.
With respect to this example, let us consider the garbled source tuple $(\tilde{X}_1, \tilde{X}_2)$, defined by $\tilde{X}_i = \delta_{1}(X_i) \, X_i$.
This new distribution is given by Table~\ref{table:leading_example}B.
We can imagine this distribution as the worst-case failure mode in every instance, while still guaranteeing that one source will be available.
Any decision-maker who has access to $\tilde{X}_1,\tilde{X}_2$ always has access to the ground-truth of either $X_1$ or $X_2$, and thus, we argue, the information redundant to both of them.
As we see, $I(T; \tilde{X}_1,\tilde{X}_2) = 0$.
Thus, it stands to reason that $R=0$.

Our proposed redundancy function $\Ift$ will be defined using such garblings, which we think of as source-fallible instantiations of the underlying ground-truth system.
First, we review the partial information decomposition for the general case.

%% file: Sections/pid.tex
\section{Partial Information Decomposition}
\label{sec:PID}

PID concerns the system of random variables composed of a target $T$ and a collection of source or predictor variables
$\bm{X} = (X_i)_{i=1}^n$.
In this work we consider all variables, including the target $T$ and predictors $X_i$, to be discrete, with finite alphabets $\mathcal{A}_T, \mathcal{A}_{X_i} \subset \mathbb{N}$.
Note that we exclude `0' from these alphabets, as we will be reserving it for source failure later.
This collection of variables forms the base system that represents the `ground truth' of both sensor and target variables, assuming no failures.
We refer to this as the base system, to distinguish it from source-fallible systems later.

\begin{definition}[Base System]
    For a given $n \in \mathbb{N}$, a predictor-target \textbf{base system} is the collection of $n+1$ finite random variables $\mathfrak{s} = (\bm{X}, T) = (X_1, ..., X_n, T)$, characterized by the joint pmf $p_{\mathfrak{s}} = p_{\bm{X},T}$.
    For any index set $I \subseteq [n]$, we may denote the subsystem $\mathfrak{s}_{I} = (\bm{X}_I, T) = (X_{i_1}, ... X_{i_m}, T)$. 
\end{definition}

The PID framework was put forward as a method of decomposing the information in such systems --- that is, decomposing the mutual information $I(T;\bm{X})$ into constitutive parts attributable to each $X_i$ and their combinations \cite{williams2010nonnegative}.
This attribution does not follow obviously from pairwise informations $I(T;X_1)$, since information is non-additive in general.

In the introduction, we introduced the bivariate PID in (\ref{eqs:bivariate_PID}a-c), in which $I(T;X_1, X_2)$ is decomposed into atoms $R+U_1+U_2+S$.
These atoms are typically given the following names and interpretations \cite{lizier2018information}
\begin{itemize}
    \item \textbf{Redundant} or \textbf{shared information ($\bm{R}$)}. The information regarding $T$ common to both predictors, and in some sense available from either.
    \item \textbf{Unique information ($\bm{U_i}$)}. The information re: $T$ available from only one of the variables, independently of the other.
    \item \textbf{Synergistic} or \textbf{complementary information ($\bm{S}$)}. The information re: $T$ that is only available when both variables are known.
\end{itemize}
PID extends for $n\geq 2$, though the number of distinct information atoms scales superexponentially.
Moreover, the interpretation of higher-order information atoms is not straight-forward.
As in the bivariate case above, in which there are 3 independent constraints for 4 unknowns, the PID equations for any $n$ present us with an underdetermined system.
Thus, a definition for one of the quantities must be provided.
Typically, this is done by defining a redundancy, synergy, or unique information function, and allowing the other atoms to follow.
Ours is a redundancy function $\Ift$, and thus the recursive derivation of the full PID lattice follows exactly as it did for $\Imin$ in \cite{williams2010nonnegative}.

We will now review the construction of the PID for an arbitrary number of predictors $n$.
From here on, \textbf{sources} are
subcollections of source variables, denoted  here $\bm{X}_I \subseteq \bm{X}$ for $I \subseteq [n]$.
PID can be summarized as a framework for measuring the redundancy and synergy among collections of information sources --- specifically, their information with respect to the designated target.
Note that we will sometimes treat both the full set $\bm{X}$ and sources $\bm{X}_{I} \subseteq \bm{X}$ as sets, and other times as tuples/vectors.

For any set $A$, allow $\mathcal{P}'(A) = \mathcal{P}(A) \setminus \{ \emptyset \}$ to denote the collection of all non-empty subsets, where $\mathcal{P}$ is the standard powerset.
The collection $\mathcal{P}'(\bm{X})$ forms a partially ordered set (`poset') under inclusion, denoted $(\mathcal{P}'(\bm{X}), \subseteq)$.
We let $I(T; \bm{X}_I)$ denote source-target mutual information.
For a fixed base system $\mathfrak{s}$, $I(T; \cdot)$ can be thought of as a functional on $\mathcal{P}'(\bm{X})$ that is monotonically increasing with respect to the partial order `$\subseteq$.'
PID assigns information value to collections of sources interpretable as redundant, synergistic, and/or contributions from the underlying variables.
In order to avoid trivial cases, only the antichains of $(\mathcal{P}'(\bm{X}), \subseteq)$ are considered.
An antichain of a poset is a collection of its elements such that no two are directly comparable.
Let $\mathcal{A}(\bm{X})$ denote the antichains of $(\mathcal{P}'(\bm{X}), \subseteq)$, defined formally by 
\begin{align}
    \nonumber
    \mathcal{A}(\bm{X}) \triangleq \bigg\{ & \alpha \in \mathcal{P}'(\mathcal{P}'(\bm{X})) \; \bigg| \\
    & \forall \bm{A},\bm{A}' \in \alpha, \; \bm{A} \neq \bm{A}' \Rightarrow \bm{A} \not\subseteq \bm{A} '\bigg\}
\end{align}
The idea here is that any $\alpha \in \mathcal{A}(\bm{X})$ is a collection of sources within which one source will never trivially dominate another by inclusion.
For a counter-example, consider that in any system $\mathfrak{s}$, $(X_1,X_2)$ will always provide no less information pointwise than $X_1$ alone, and so any sensible redundancy function will identify the redundancy of the non-antichain $\{ \{X_1, X_2\}, \{X_1\} \}$ with $I(T;X_1)$.
There are multiple partial orderings that may be assigned to the antichains themselves \cite{crampton2000two}.
In PID, the following is used \cite{williams2010nonnegative}:
\begin{align}
\label{eq:PID_lattice_order}
    \forall \alpha, \beta \in \mathcal{A}(\bm{X}), \quad \alpha \preceq \beta \Leftrightarrow \forall \bm{B} \in \beta, \exists \bm{A} \in \alpha, \bm{A} \subseteq \bm{B}
\end{align}
Although every antichain $\alpha$ is a set of sets of indexed predictors, we will typically exclude the outer brackets from notation where possible, i.e. $\{X_1\}\{X_2\}$ in place of $\{\{X_1\},\{X_2 \}\}$.
Moreover, we will let $\ind(\alpha)$ denote the collection of index sets for $\alpha$, i.e. $\ind(\alpha) = \{ I_j \}_j$ when $\alpha = \{\bm{X}_{I_j}\}_j$.

A \textbf{PID function} $\Pi$ assigns an information value to each $\alpha$ in $\mathcal{A}(\bm{X})$, while satisfying the \textbf{PID equations}:
\begin{equation}
    \forall \bm{X}_I \in \mathcal{P}'(\bm{X}), \quad
    I(T;\bm{X}_I) = \sum_{\alpha \preceq \{X_I\}} \Pi(\alpha)
\end{equation}
In the bivariate case where $\bm{X} = (X_1, X_2)$, these resolve to Eqs.~(\ref{eqs:bivariate_PID}a-c), which in our more general notation take the form:
\vspace{-12pt}
\begin{subequations}
\begin{align}
    I(T;X_1) & = \Pi(\{ X_1\} \{X_2\}) + \Pi(\{ X_1\}) \\
    I(T;X_2) &= \Pi(\{ X_1\} \{X_2\}) + \Pi(\{ X_2\}) \\
    \nonumber I(T;X_1, X_2) & = \Pi(\{ X_1\} \{X_2\}) + \Pi(\{ X_1\}) \\
    & \quad  + \Pi(\{ X_2\}) + \Pi(\{ X_1, X_2\})
\end{align}
\end{subequations}
We refer the reader to \cite[Figs.~2-3]{williams2010nonnegative} for visualizations of the PID lattice for $n \leq 3$.

The original \cite{williams2010nonnegative} and most common approach to defining $\Pi$ is by first defining a redundancy function $\Ired$ on $\mathcal{A}(\bm{X})$, which measures the redundant information among all the sources in $\alpha$.
The PID function $\Pi$ is then derived as the M\"{o}bius inverse of $\Ired$:
\begin{equation}
    \Pi(\alpha) = \Ired(\alpha) - \sum_{\beta \preceq \alpha} \Ired(\beta)
\end{equation}

It has often be considered desirable for a proposed PID to satisfy a set of properties first mentioned in \cite{williams2010nonnegative}, which have sometimes been referred to as the Williams-Beer (WB) axioms \cite{finn2018}.
\begin{enumerate}
    \item \textbf{Symmetry}. $\Ired(\bm{X}_{I_1}, ..., \bm{X}_{I_k})
    = \Ired(\bm{X}_{I_{\sigma(1)}}, ..., \bm{X}_{I_{\sigma(k)}})$ for any permutation $\sigma$.
    \item \textbf{Self-Redundancy.} $\Ired(\bm{X}_I) = I(T;\bm{X}_I)$
    \item \textbf{Monotonicity.} $\alpha \preceq \beta \Rightarrow \Ired(\alpha) \leq \Ired(\beta)$
\end{enumerate}
The symmetry axiom is usually trivial, and we ignore it in this work.
The monotonicity axiom is typically stated in its weaker form $\alpha \subseteq \beta \Rightarrow \Ired(\alpha) \leq \Ired(\beta)$,
since inclusion is a weaker ordering for antichains than that from (\ref{eq:PID_lattice_order}).
However, the stronger form here was demonstrated in the original PID \cite{williams2010nonnegative}, and follows naturally from our main theorem in the next section.

%% file: Sections/ift.tex
\section{A Measure of Redundant Information in Source-Fallible Systems}
\label{sec:Ift}

In Sec.~\ref{sec:example}, we gave an example of the kind of redundancy we aim to quantify.
In the fault tolerance literature, redundancy is provided as the number of nodes $n$ we would need in order to maintain system functionality if we allow for up to $\ell$ failures \cite{rullo2019redundancy}.
Here, we ask a slightly different question.
Suppose we have $n$ sources of information $X_1, ..., X_n$, and we want to know whether we may tolerate $\ell$ source failures while still having enough information regarding $T$.
That determination may be made by considering $I(T;\tilde{\bm{X}})$, for some $\tilde{\bm{X}} = (\tilde{X}_i)_i$ analagous to that in Table~\ref{table:leading_example}B, earlier.
We begin by formally defining such source-fallible instantiations for a system of interest.

\begin{definition}[Source-Fallible System]
\label{defn:sfs}
Let a base system $\baseSystem$ be given.
A source-fallible instantiation of this system will be given by the binary sensor failure variables $\bm{F} = (F_1, ..., F_n)$, $\mathcal{A}_{\bm{F}}=\{0,1\}^n$,
which are fully characterized by the conditional distribution $p_{\bm{F} | \bm{X}, T}$.
The event $F_i = 0$ is interpreted as a `failure' of the sensor for $X_i$.
From a given $\bm{F}$, we define the induced source-fallible system (SFS)
as the tuple of random variables $\sfs = ( \tilde{\bm{X}}, T) = ( \tilde{\bm{X}}(\bm{F},\bm{X}), T)$, where $T$ is the same target variable as in $\baseSystem$ and $\tilde{\bm{X}}$ is given by
\begin{align}
\label{eq:SFS}
\tilde{X}_i = \begin{cases}
    0, & F_i = 0\\
    X_i, & F_i = 1\\
\end{cases}
\end{align}
The collection of all SFS's associated to base system $\mathfrak{s}$ will be denoted $\mathfrak{F} = \mathfrak{F}(\mathfrak{s})$.
For any subsystem $\mathfrak{s}_I \subset \mathfrak{s}$, we may use the short-hand $\mathfrak{F}_I = \mathfrak{F}(\mathfrak{s}_I)$, which is similarly  the set of subsetted sensor-fallible vectors $\tilde{\bm{X}}_I$, satisfying (\ref{eq:SFS}) for each $i \in I$.
\end{definition}

The idea here is that our notion of redundancy deals with censored copies of our predictors, where for fixed $\baseSystem$ there's a one-to-one correspondence between the failure distribution $\bm{F}$ and the censored source vector $\bm{\tilde{X}}$.
Clearly, the least informative $\tilde{\bm{X}}$ in $\SFS$ will be given by $\bm{F} \equiv \bm{0}$, i.e. when all sources fail all the time. Then $\tilde{\bm{X}} \equiv \bm{0}$.
The SFS's we're interested in, though, are those that have limited failures while still providing one of the sources of information, as in Table~\ref{table:leading_example}B.
We now have the notation to make this condition precise in the general case.
We will allow for a fairly unconstrained distribution of source failures $\bm{F}$, i.e. we need not assume $\bm{F}$ has any particular independence with respect to $T$ and/or $\bm{X}$.

\begin{definition}
    \label{defn:sfsrs}
    Let a system $\baseSystem$ be given. We say that an SFS $\sfs(\bm{F})$ \textbf{redundantly satisfies} a source antichain $\alpha \in \mathcal{A}(\bm{X})$ if
    \begin{equation}
    \label{eq:defn.redSat}
        P \left( \bigvee_{I \in \ind(\alpha)} \bm{F}_I = \bm{1}_{|I|} \right) = 1
    \end{equation}
    where $\bm{1}_{|I|}$ is the length $|I|$ vector of ones.
    The collection of all $\sfs \in \SFS(\baseSystem)$ redundantly satisfying $\alpha$ are denoted $\SFSRS(\alpha).$
\end{definition}
Intuitively, (\ref{eq:defn.redSat}) means that always at least one source $X_I \in \alpha$ is fully available from observing $\tilde{\bm{X}}$.
\textbf{We posit that the information redundant among the sources in an antichain should be available from any SFS that redundantly satisfies that antichain.}
From this postulate, the definition that should follow for redundant information becomes clear.
We merely minimize over all such SFS's redundantly satisfying the given antichain.

\begin{definition}
    \label{defn:Ift}
    Let a base system $\baseSystem$ be given. For any source antichain $\alpha \in \mathcal{A}(\bm{X})$, the $\Ift$ redundancy function is defined by
    \begin{align}
    \label{eq:defn.Ift}
        \Ift(\alpha) = \min_{(\tilde{\bm{X}},T) \in \SFSRS(\alpha)} I(T; \tilde{\bm{X}})
    \end{align}
\end{definition}

There is one aspect of this definition that may seem a little strange.
For any antichain, we are minimizing over SFS's constructed from the full system $(X_1, ..., X_n, T)$.
For the redundancy of $\alpha = \{ X_1\} \{X_2\}$, surely we should only be concerned with $(X_1, X_2, T)$.
In fact, the minimum in (\ref{eq:defn.Ift}) for such a case can be achieved for an SFS where $F_i \equiv 0$ for $i>2$,.
This is indistinguishable from considering only the system $(X_1,X_2, T)$ from the start.
Although we retain Def.~\ref{defn:Ift} as our preferred formulation, in the following two propositions, we make precise the equivalence of our definition to this reduced form.

\begin{proposition}
    \label{prop:embedding}
    For any base system $\baseSystem$ and subsystem $\baseSystem_I \subset [n]$, there is an injective embedding  $g_I: \SFS_I \hookrightarrow \SFS$,
    given by
    \begin{gather}
        \nonumber g_I: \sfs_I = (\tilde{\bm{X}}_I,T) \mapsto g_I(\sfs_{I}) = (\tilde{\bm{Y}},T) \\
        \label{eq:embedding_map}
        \tilde Y_i = \begin{cases}
            \tilde X_i, & i \in I \\
            0, & i \in [n] \setminus I
        \end{cases}
    \end{gather}
    This embedding is information preserving, in the sense that
    \begin{equation}
    \label{eq:embedding.info_equiv}
        I(T;\tilde{\bm{X}}_I) = I(T;\tilde{\bm{Y}})
    \end{equation}
\end{proposition}
\begin{proof}
The injectivity of $g_I$ is straightforward from its definition. For (\ref{eq:embedding.info_equiv}), we observe that $I(T;\bm X_I) = I(T;\bm Y_I)$ and $I(T;\bm Y_{\bar{I}} | \bm Y_I) \leq H(\bm Y_{\bar{I}} |\bm Y_I) = 0$ (as $\bm Y_{\bar{I}}$ is constant).
Thus, (\ref{eq:embedding.info_equiv}) follows from the chain rule $I(T;\bm{Y}) = I(T;\bm{Y}_I) + I(T;\bm{Y}_{\bar{I}} | \bm{Y}_I)$.
\end{proof}

We may thus identify $\SFS_I$ with its image $g_I(\SFS_I)$ as a subset of $\SFS$.
The following proposition shows that we may compute $\Ift(\alpha)$ within the reduced space $\SFS_I$ for the smallest $\bm{X}_I$ capturing all sources in $\alpha$.

\begin{proposition}
    For any $\alpha \in \mathcal{A}(\bm{X})$, we let $\baseSystem_{\alpha} = \baseSystem_{\cup \ind(\alpha)}$ be the subsystem formed by the union of the sources in $\alpha$.
    The sensor-fallible realizations of this subsystem are similarly denoted $\SFS_{\alpha}$.
    Using the embedding $\SFS_{\alpha} \hookrightarrow \SFS$ from Prop.~\ref{prop:embedding},
    we denote the intersection of $\SFS_{\alpha}$ with antichains redundantly satisfying $\alpha$ by:
    \begin{equation}
        \SFSRS_{\alpha}(\alpha) = \SFSRS(\alpha) \cap \SFS_{\alpha}.
    \end{equation}
    Then the $\Ift$ function from Def.~\ref{defn:Ift} is equivalently defined by:
    \begin{align}
    \label{eq:defn.Ift.subsystem}
        \Ift(\alpha) = \min_{(\tilde{\bm{X}},T) \in \SFSRS_{\alpha}(\alpha)} I(T; \tilde{\bm{X}}).
    \end{align}
\end{proposition}
\begin{proof}
    Let $(\tilde{\bm{X}},T) \in \SFSRS(\alpha)$ be given.
    Set $I^{\cup} = \cup \ind(\alpha)$.
    Define $(\tilde{\bm{Y}},T)$ by $\tilde{Y}_i \equiv \tilde{X}_i$ for $i \in I^{\cup}$, and $\tilde{Y}_i \equiv 0$ otherwise.
    Clearly, $(\tilde{\bm{Y}},T) \in \SFSRS_{\alpha}(\alpha)$, satisfying (\ref{eq:defn.redSat}) for $\alpha$.
    Moreover, $I(T;\tilde{\bm{X}}) \geq I(T;\tilde{\bm{X}}_{I^{\cup}}) = I(T;\tilde{\bm{Y}})$.
    Thus, since we can find such a $(\tilde{\bm{Y}},T) \in \SFSRS_{\alpha}(\alpha)$ for any $(\tilde{\bm{X}},T) \in \SFSRS(\alpha)$, it follows that the minimum in (\ref{eq:defn.Ift}) is equal to that in (\ref{eq:defn.Ift.subsystem}).
\end{proof}

From the preceding two propositions, the demonstration of the self-redundancy PID axiom for $\Ift$ reduces to one line.
\begin{proposition}[Self-Redundancy Axiom]
    \label{prop:SRaxiom}
    The $\Ift$ function satisfies the self-redundancy axiom from Sec.~\ref{sec:PID} --- that is, it is an extension of mutual information in the sense that
    \begin{equation}
        \label{eq:prop.SR}
        \Ift(\{ X_I \}) = I(T;X_I)
    \end{equation}
    for all $\bm{X}_I$ 
\end{proposition}
\begin{proof}
    Since $I = \cup \ind(\{ X_I \})$, $\SFSRS_{\alpha}(\alpha)$ is a singleton.
    The only element is $(X_I,T)$, and so (\ref{eq:prop.SR}) follows from (\ref{eq:defn.Ift.subsystem}).
\end{proof}

Our main result in this section elucidates the structure of the $\Ift$ redundancy function, including its satisfaction of the monotonicity PID axiom, by establishing a close correspondence between source antichains $\alpha \in \mathcal{A}(\bm{X})$, on the one hand, and these collections of source-fallible systems $\SFSRS(\alpha) \subset \SFS(\baseSystem)$ that we have been examining.
Let $\mathcal{F}(\baseSystem) = \{ \SFSRS(\alpha) \}_{\alpha \in \mathcal{A}(\bm{X})}$.
Insofar as each family $\SFSRS(\alpha)$ is a subset of $\SFS$, it follows that we can consider $\mathcal{F}(\baseSystem)$ as a finite poset under inclusion, i.e. $(\mathcal{F}(\baseSystem), \subseteq)$.
As we will now demonstrate, this poset is the mirror image of the PID lattice $(\mathcal{A}(\bm{X}), \preceq)$.

\begin{thm}
For any base system $\baseSystem = (\bm{X}, T)$, the lattices $(\mathcal{A}(\bm{X}), \preceq)$ and $(\mathcal{F}(\baseSystem), \subseteq)$ are anti-isomorphic.
Namely, the map $\alpha \mapsto \SFSRS(\alpha)$ is an order-reversing bijection, in the sense that
\begin{equation}
    \label{eq:poset_antiisomorphism}
    \alpha \preceq \beta \Leftrightarrow \SFSRS(\beta) \subseteq \SFSRS(\alpha) 
\end{equation}
\end{thm}
\begin{proof}
    Suppose $\alpha \preceq \beta$.
    Let $(\tilde{\bm{X}},T) \in \SFSRS(\beta)$ be given, induced from $\bm{F}$.
    We will show that $(\tilde{\bm{X}},T) \in \SFSRS(\alpha)$.
    Let $\bm{f}$ be a given realization of $\bm{F}$.
    We then know by Def.~\ref{defn:sfsrs} that for some index set $I \in \ind(\beta)$, for source $\bm{X}_I \in \beta$, $\bm{f}_I = \bm{1}$.
    In other words, for the event $\bm{F}=\bm{f}$, the source $\bm{X}_I$ is available.
    By (\ref{eq:PID_lattice_order}), we know that there is some $\bm{X}_J \in \alpha$ such that $\bm{X}_J \subset \bm{X}_I$, i.e. $J \subseteq I$ for some $J \in \ind(\alpha)$.
    Therefore, $\bm{f}_J = \bm{1}$ and $X_J$ is available. Since our choice of $\bm{f}$ was arbitrary, it follows that (\ref{eq:defn.redSat}) holds for $\alpha$ as well as $\beta$, and so $(\tilde{\bm{X}},T)$ redundantly satisfies $\alpha$.

    Now, suppose instead $\alpha \not \preceq \beta$.
    Then there is some $X_I \in \beta$ such that for all $X_J \in \alpha$, $X_J \not\subseteq X_I$, i.e $J \not\subseteq I$.
    This means for each $J \in \ind(\alpha)$, there is some $j^\star \in J \cap \bar{I}$.
    Let $\bm{F}$ be the constant binary vector where $F_I = \bm{1}$ and $F_{\bar{I}} = \bm{0}$.
    This induces the SFS $(\tilde{\bm{X}}, T)$ which redundantly satisfies $\beta$, since $P(\bm{F}_I) = 1$.
    However, $P(\bm{F}_{J} = \bm{1}) = 0$ for every $J \in \ind(\alpha)$, since $F_{j^\star} \equiv 0$.
    Thus, $(\tilde{\bm{X}}, T) \in \SFSRS(\beta) \setminus \SFSRS(\alpha)$.
\end{proof}

An immediate consequence of this result is the monotonicity PID axiom.
As one adds sources to an antichain, the set of fallible instantiations of the system grows: there are more ways for things to go wrong.
Conversely, by removing sources or concatenating them, we reduce that range of possibilities.

\begin{corollary}[Monotonicity Axiom]
The $\Ift$ function is (monotone) increasing on the PID lattice $(\mathcal{A}(\bm{X}), \preceq)$, i.e.
\begin{equation}
    \label{eq:cor.monotonicity}
    \alpha \preceq \beta \Rightarrow \Ift(\alpha) \leq \Ift(\beta)
\end{equation}
\end{corollary}
\begin{proof}
    By (\ref{eq:poset_antiisomorphism}), we have 
    \begin{align}
    \nonumber
        \alpha \preceq \beta & \Leftrightarrow \quad \SFSRS(\alpha) \supseteq \SFSRS(\beta ; \baseSystem) \\
        & \Rightarrow \quad \min_{\sfs \in \SFSRS(\alpha ; \baseSystem)} f(\sfs) \leq \min_{\sfs \in \SFSRS(\beta ; \baseSystem)} f(\sfs)
    \end{align}
    for any function $f$.
    (\ref{eq:cor.monotonicity}) follows from (\ref{eq:defn.Ift}).
\end{proof}

To conclude, let us revisit the scenario we introduced at the start of this section.
Given $n$ sources of information (predictor variables) for a system regarding its target $T$, we asked how that system might tolerate up to $\ell$ failures at a time.
Let us further add that these failures may be arbitrarily (even adversarially) distributed.
Let $\mathcal{P}_{n - \ell}'(\bm{X})$ be the set of all sources $\bm{X}_I$ with $|I| = n - \ell$.
There will be ${n \choose \ell}$ of these sources.
Then, letting $\alpha_{\ell} = \mathcal{P}_{n - \ell}'(\bm{X})$, we can answer that the expected amount of information we can guarantee regarding $T$ is given by $\Ift(\alpha)$, as defined in this section.
Moreover, as a special case of the monotonicity axiom demonstrated above, we have that
\begin{align}
    \nonumber 0 &= \Ift(\alpha_n) \leq \Ift(\alpha_{n-1}) \leq ... \\
    &\leq \Ift(\alpha_{1}) \leq \Ift(\alpha_{0}) = I(T; \bm{X})
\end{align}
where we extended $\Ift$ naturally to take $\alpha_n = \emptyset$ as an argument.